\documentclass[letterpaper, 10 pt, conference]{ieeeconf}  

\IEEEoverridecommandlockouts                              
\usepackage{bm}
\usepackage{amsmath,amssymb,amsfonts}
\usepackage{subfigure}
\usepackage{graphicx}
\usepackage{cite}
\newtheorem{assumption}{Assumption}
\newtheorem{lemma}{Lemma}
\newtheorem{theorem}{Theorem}

\newtheorem{example}{Example}
\usepackage[T1]{fontenc}

\title{\LARGE \bf
Non-Bayesian Social Learning with Multiview Observations
}

\author{Dongyan Sui, Weichen Cao, Stefan Vlaski, Chun Guan, Siyang Leng
\thanks{D.Y.S., W.C.C., C.G., and S.Y.L. are with Academy for Engineering and Technology, Fudan University, Shanghai 200433, China. S.Y.L. is with Research Institute of Intelligent Complex Systems, Fudan University, Shanghai 200433, China. S.V. is with Department of Electrical and Electronic Engineering, Imperial College London, UK. Corresponding e-mails: \texttt{\{chunguan, syleng\}@fudan.edu.cn}.}
\thanks{S.Y.L. is supported by the National Natural Science Foundation of China (No. 12101133).}%
}

\begin{document}

\maketitle
\thispagestyle{empty}
\pagestyle{empty}

\begin{abstract}
Non-Bayesian social learning enables multiple agents to conduct networked signal and information processing through observing environmental signals and information aggregating. Traditional non-Bayesian social learning models only consider single signals, limiting their applications in scenarios where multiple viewpoints of information are available. In this work, we exploit, in the information aggregation step, the independently learned results from observations taken from multiple viewpoints and propose a novel non-Bayesian social learning model for scenarios with multiview observations. We prove the convergence of the model under traditional assumptions and provide convergence conditions for the algorithm in the presence of misleading signals. Through theoretical analyses and numerical experiments, we validate the strong reliability and robustness of the proposed algorithm, showcasing its potential for real-world applications.
\end{abstract}

\section{Introduction}
Networked signal and information processing\cite{4663899,4749425,6197748,6172233,10188495} refers to the collaborative processing of information and signals among a network of distributed agents. This approach leverages the collective capabilities of interconnected devices to perform tasks like decision-making, inference, and learning more efficiently than isolated systems. The significance of networked information processing lies in its ability to enhance performance through cooperation, offering advantages such as improved scalability, resilience, and resource efficiency. It is particularly relevant in applications like sensor networks, distributed control systems, and collaborative robotics.

Non-Bayesian social learning\cite{jadbabaie2012non,shahrampour2013exponentially,8359193,7040286,8815195,7891016,7172262,7851025,8815382,10383999,9909947,9413642,bhotto2018non,9011362,hare2020non} offers a novel framework for networked signal and information processing, enabling a distributed way for agents with limited rationality and diverse sensing capabilities to infer collectively over a network. Agents process streams of incomplete data based on the underlying true state of the world, using network communications to form beliefs about various possible hypotheses and make an estimate of the underlying true state. This collaborative mechanism, which integrates neighbors' insights with fresh individual data, fosters a scalable method of learning without prior knowledge of network structure or historical data.

Various social learning models have been proposed, including aggregation methods such as linear averages\cite{jadbabaie2012non,shahrampour2013exponentially}, geometric averages\cite{8359193,7040286}, and the minimum operator\cite{8815195}. These rules can be applied to different network structures, including undirected/directed, time-varying\cite{7891016,7172262}, weakly-connected graphs\cite{7851025}, and higher-order topology\cite{chen2023distributed}, as well as to agents with growing self-confidence\cite{8815382} and heterogeneous stubbornness parameters\cite{10383999}, disparate hypothesis\cite{9909947}, under inferential attacks\cite{9413642}, and in adversarial conditions\cite{bhotto2018non}. Research has also explored learning under uncertain likelihood models and performance against malicious agents\cite{9011362,hare2020non}. All of these models offer theoretical assurances that, over time, agents can collectively learn the underlying true state of the world.

In previous non-Bayesian social learning models, the agents receive signals from the environment and cooperatively infer the underlying state based on their \textit{a priori} knowledge of the signals. However, in practice, the group could perceive different features of the environment from various perspectives. For example, individuals could judge the species of trees based on the characteristics of both leaves and trunks; customers always infer the quality of a target product by observing the quality of other products from the same brand. Traditional methods cope with such situations by integrating these multiple viewpoints into one single signal, making it challenging to determine the likelihood functions of agents with the integrated signal due to the requirement of substantial data to assess the independence or correlation among multiple signals.

In this work, we propose a novel non-Bayesian social learning algorithm based on multiview observations. Our algorithm allows the group to learn independently from multiple signals and achieves interaction among multiview observations during the information aggregation process. Similarly to previous methods, we prove the correct convergence of our proposed algorithm under traditional assumptions. Additionally, we provide convergence conditions based on the presence of misleading signals. Numerical experiments validate the effectiveness of our theoretical analysis, and we showcase the robust fault-tolerance capability of our proposed algorithm in the task of multi-agent collaborative localization.

The remaining part of this paper is organized as follows: Section II provides a full description of the problem settings and introduces our learning strategies. Section III presents sufficient assumptions/lemmas and proves the convergence of the proposed algorithm. Section IV provides extensive numerical examples illustrating the theoretical results and demonstrating the effectiveness and applicability of the algorithm. The findings are concluded in Section V with possible future works.

\section{Preliminaries and the model}

\subsection{Problem formulation}

Consider a group of $n$ agents, collectively trying to reveal the underlying true state of nature, denoted as $\theta^*$, from a finite set of hypotheses $\Theta=\{\theta_1,\theta_2,\cdots,\theta_m\}$. At each time step $t=1,2,\cdots$, agent $i$ obtains $p$ types of observations $\left\{s^l_{i,t}\right\}_{l=1}^p$, which may come from multiple perspectives or represent different features of the true state. Each element $s^l_{i,t}$ is the realization of an environmental random variable $\bm{S}^l_{i,t}$. The set $\bm{s}^l_t=\left\{s^l_{1,t},s^l_{2,t},\cdots,s^l_{n,t}\right\}$ represents the actual observations made by all agents from signal type $l$ at time $t$, generated according to the likelihood function $\bm{f}^l(\cdot)$ associated with the underlying true state $\theta^*$. The set $\tilde{\bm{s}}_t=\left\{\bm{s}^1_t,\cdots,\bm{s}^p_t\right\}$ and $\tilde{\bm{f}}=\bm{f}^1\times\cdots\times\bm{f}^p$. Each $\bm{S}^l_{i,t}$ has its individual observation space $S^l_i$ and is i.i.d. with respect to $t$.

The signal structure for agent $i$ with signal type $l$ and possible state $\theta$ is described by a probability distribution $\ell^l_i(\cdot|\theta)$. In these settings, $\ell^l_i(s^l_{i,t}|\theta)$ indicates the likelihood of agent $i$ observing type $l$ signal $s^l_{i,t}$ at time $t$ when it believes $\theta$ is the underlying true state.

The agents interact in a networked fashion, which is usually modelled by a directed graph $\mathcal{G}=(\mathcal{V},\mathcal{E})$. $\mathcal{V}=\{1,2,\cdots,n\}$ is the set of vertices representing the $n$ agents, and $\mathcal{E}=\{(i,j)|\textrm{agent~}j\textrm{~can receive information from agent~}i\}$ is the set of directed edges. We denote $A=(a_{ij})_{n\times n}$ as the weight matrix of $\mathcal{G}$, which is assumed to be row-stochastic, i.e., $\sum\limits_{j=1}^na_{ij}=1, \forall i=1,\cdots,n$, and $a_{ij}>0$ if $(j,i)\in\mathcal{E}$. The row-stochastic condition of $A$ ensures that all agents assign normalized weights to the information, i.e., proportions of the total, that they receive from neighbors.

The \textit{belief} of agent $i$ at time $t$ with signal type $l$ is denoted as $\mu^l_{i,t}$, which is a probability distribution over the set of states $\Theta$, i.e., $\sum\limits_{k=1}^m\mu^l_{i,t}(\theta_k)=1$, $\forall i=1,\cdots,n, \forall l=1,\cdots,p$, and $\forall t=0,1,\cdots$. Here $\mu^l_{i,0}$ represents the \textit{initial belief} of agent $i$ with signal type $l$.

Define a probability triple $(\Omega,\mathcal{F},\mathbb{P}^*)$, where $\Omega=\{\omega|\omega=(\tilde{\bm{s}}_1,\tilde{\bm{s}}_2,\cdots)\}$, $\mathcal{F}$ is the $\sigma$-algebra generated by the observations, and $\mathbb{P}^*$ is the probability measure induced by paths in $\Omega$, i.e., $\mathbb{P}^*=\prod\limits_{t=1}^\infty\tilde{\bm{f}}$. We use $\mathbb{E}^*[\cdot]$ to denote the expectation operator associated with measure $\mathbb{P}^*$.

In this work, we consider the following two different circumstances:

Circumstance 1: For every agent $i$ and every signal type $l$, the signal structure $\ell^l_i(\cdot|\theta^*)$ aligns with the $i$-th marginal distribution of $\bm{f}^l(\cdot)$ for all $l=1,\cdots,p$, thereby characterizing the probability distribution of $\bm{S}^l_{i,t}$. In this case, all agents' \textit{a priori} knowledge is accurate, and none of the signal types are misleading.

Circumstance 2: The condition in Circumstance 1 is not satisfied, and there may exist a signal type $l$ such that $\bm{\ell}^l(\cdot|\theta^*)$ is not the best match of the real distribution $\bm{f}^l(\cdot)$ from the group's perspective. In this case, the group may experience false learning solely based on the type $l$ signal.

The second circumstance could be quite common in practical applications, often arising from faults in signal perception or incorrect prior information due to a lack of training data.

\subsection{Social Learning Strategies with Multiple Signals}

Non-Bayesian social learning typically involves two steps for agents to update their beliefs at each time, i.e., the Bayesian update step and the aggregation of neighbors' beliefs ~\cite{qipeng2015distributed, lalitha2014social,nedic2017fast}. In the belief aggregation step, every agent shares its current belief with its one-hop neighbors, whereas in the Bayesian update step, every agent combines its prior belief with observations from the environment to form its posterior belief.

Traditionally, when dealing with tasks involving multiview observations, social learning algorithms integrate these diverse signals into a single signal and design a joint likelihood function as the signal structure. This approach demands a thorough understanding of the correlations among different viewpoints of observations, often making it challenging to achieve in practical tasks. In our work, however, we allow agents to independently perform Bayesian inference for each signal type and integrate information from multiview observations during the information aggregation process.

The algorithm we propose can be described in the following two steps:

1) Information aggregation. For each agent $i$ and signal type $l=1,\cdots,p$, we calculate the updated belief using the formula:
\begin{equation*}
\begin{aligned}
\tilde{\mu}_{i,t+1}^l&(\theta)=\\
&\frac{\exp\left(\gamma_l\sum\limits_{j=1}^na_{ij}\log\mu_{j,t}^l(\theta)+\sum\limits_{k\neq l}^p\gamma_k \log \mu_{i,t}^k(\theta)\right)}{\sum\limits_{\theta'\in\Theta}\exp\left(\gamma_l\sum\limits_{j=1}^na_{ij}\log\mu_{j,t}^l(\theta')+\sum\limits_{k\neq l}^p\gamma_k \log \mu_{i,t}^k(\theta')\right)},
\end{aligned}
\end{equation*}
where the assigned parameter $\gamma_l\in (0,1)$ for all $l=1,\cdots,p$, and $\sum\limits_{l=1}^p\gamma_l=1$.

2) Bayesian update. For each agent $i$ and signal type $l=1,\cdots,p$, the posterior belief is given by:
\begin{equation*}
\mu_{i,t+1}^l(\theta)=\frac{\tilde{\mu}^l_{i,t+1}(\theta)\ell_i^l(s^l_{i,t+1}|\theta)}{\sum\limits_{\theta'\in\Theta}\tilde{\mu}^l_{i,t+1}(\theta')\ell_i^l(s^l_{i,t+1}|\theta')}.
\end{equation*}

\section{Assumptions and results}

As widely discussed in previous works of social learning, we care about the convergence of the algorithms as well as the rate of convergence. The following assumptions are required to ensure the convergence of the proposed social learning strategies:
\begin{assumption}[Communication network]
The graph $\mathcal{G}=(\mathcal{V},\mathcal{E})$ and its weight matrix $A$ satisfy that:

a) The graph is strongly-connected;

b) $A$ has at least one positive diagonal entry.
\end{assumption}

Assumption 1 ensures that $A$ is the transition matrix of an irreducible, aperiodic Markov chain with finite states. We recall the following lemma~\cite{hoel1986introduction}:
\begin{lemma}
If a Markov chain with finite states is irreducible, then it has a unique stationary distribution $\pi$. Let $A$ be the transition matrix of the Markov chain and further suppose it is aperiodic, then we have $\lim\limits_{k\rightarrow\infty}[A^k]_{ij}=\pi_j$, for $1\le i,j\le n$.
\end{lemma}

The stationary distribution $\pi$ can be interpreted as the normalized left eigenvector of $A$ corresponding to eigenvalue $1$, known as the \textit{eigenvector centrality} in related literature. The Perron-Frobenius theorem ensures that all components of $\pi$ are strictly positive.

\begin{assumption}[Belief and signal structure]
Every agent $i=1,\cdots,n$ in the group satisfies:

a) It has positive initial beliefs on all states regarding all types of signals, i.e., $\mu^l_{i,0}(\theta)>0$ for all $l=1,\cdots,p$ and $\theta\in\Theta$;

 b) The logarithms of its signal structures are integrable, i.e., $\mathbb{E}^*\left[|\log\ell^l_i(s^l_i|\theta)|\right]<\infty$ for all $l=1,\cdots,p$, $s^l_i\in S^l_i$, and $\theta\in\Theta$.
\end{assumption}

Assumption 2a) is imposed to ensure the well-definedness of $\log\mu^l_{i,t}(\cdot)$. Meanwhile, Assumption 2b) guarantees that $\log\ell^l_i(s^l_i|\theta)$ is real-valued almost surely\cite{cinlar2011probability}. In practical scenarios where the signal structures of the agents are Gaussian, Assumption 2b) holds naturally since Gaussian random variables are square integrable.

Two states, $\theta_j$ and $\theta_k$, are called \textit{observationally equivalent} with signal type $l$ for agent $i$ if $\ell_i^l(s_i^l|\theta_j)=\ell_i^l(s_i^l|\theta_k)$ for all $s^l_i\in S^l_i$, in which case the agent can not distinguish between these states using its own information obtained from type $l$ signal. The true state is called \textit{globally identifiable} if the set $\hat{\Theta}=\bigcap\limits_{l=1}^p\bigcap\limits_{i=1}^n\hat{\Theta}_i^l$ has only one element $\theta^*$, where $\hat{\Theta}_i^l=\{\theta\in\Theta|\ell^l_i(s_i|\theta)=\ell^l_i(s^l_i|\theta^*), \forall s^l_i\in S^l_i\}$. Intuitively, if a state $\theta'$ is observationally equivalent to $\theta^*$ with all types of signals for all agents, i.e., $\hat{\Theta}=\{\theta^*,\theta'\}$, then the two states are indistinguishable from the view of all agents, and they can not collectively learn the underlying true state.

To ensure the convergence of groups' beliefs on the true state, we introduce the following assumption:

\begin{assumption}[Globally identifiable]
The true state $\theta^*$ is globally identifiable.
\end{assumption}

Under this assumption, for all $\theta\neq\theta^*$, there exists at least one agent $i$ and a signal type $l$ such that $D_{\textrm{KL}}(\ell^l_i(\cdot|\theta^*)\parallel\ell^l_i(\cdot|\theta))$ is strictly positive, where $D_{\textrm{KL}}(P\parallel Q)$ represents the Kullback-Leibler divergence between two probability distributions $P$ and $Q$.

Denote in the following that $K^l_i(\theta^*,\theta)=D_{\textrm{KL}}(f^l_i(\cdot)\parallel\ell^l_i(\cdot|\theta^*))-D_{\textrm{KL}}(f^l_i(\cdot)\parallel\ell^l_i(\cdot|\theta))$. Its positivity or negativity depends on whether, from the perspective of agent $i$, state $\theta$ or $\theta^*$ is more likely to be the underlying true state. Notice that under Circumstance 1, $K_i^l(\theta^*,\theta)=-D_{\textrm{KL}}(\ell^l_i(\cdot|\theta^*)\parallel\ell^l_i(\cdot|\theta))$. Now we can state the main results describing the correct convergence of the proposed strategy.

\begin{theorem}
Under Circumstance 1 and Assumptions 1, 2 and 3, the proposed social learning strategy satisfies:
\begin{equation*}
\lim_{t\rightarrow\infty}\mu^l_{i,t}(\theta^*)=1,\quad\mathbb{P}^*{\rm-a.s.},\quad\forall 1\le i\le n, 1\le l\le p.
\end{equation*}
\end{theorem}
\begin{proof}
For each agent $i$, signal type $l$, and $\theta\neq\theta^*$, we have
\begin{equation*}
\begin{aligned}
\log\frac{\mu^l_{i,t+1}(\theta)}{\mu^l_{i,t+1}(\theta^*)}=&\gamma_l\sum_{j=1}^na_{ij}\log\frac{\mu^l_{j,t}(\theta)}{\mu^l_{j,t}(\theta^*)}+\sum_{k\neq l}^p\gamma_k\log\frac{\mu^k_{i,t}(\theta)}{\mu^k_{i,t}(\theta^*)}\\&+\log\frac{\ell_i^l(s^l_{i,t+1}|\theta)}{\ell^l_i(s^l_{i,t+1}|\theta^*)}.
\end{aligned}
\end{equation*}
By denoting $\nu^l_{i,t+1}(\theta)=\log\frac{\mu^l_{i,t+1}(\theta)}{\mu^l_{i,t+1}(\theta^*)}$ and $L^l_{i,t+1}(\theta)=\log\frac{\ell^l_i(s^l_{i,t+1}|\theta)}{\ell^l_i(s^l_{i,t+1}|\theta^*)}$, the above equation simplifies to
\begin{equation}\label{1}
\nu^l_{i,t+1}(\theta)=\gamma_l\sum_{j=1}^na_{ij}\nu^l_{j,t}(\theta)+\sum_{k\neq l}^p\gamma_k\nu^k_{i,t}(\theta)+L^l_{i,t+1}(\theta).
\end{equation}
Define the $n$-dimensional column vector $\bm{\nu}^l_t(\theta)=\left(\nu^l_{1,t}(\theta),\cdots,\nu^l_{n,t}(\theta)\right)^\top$ for each $l=1,\cdots,p$ and the $np$-dimensional column vector $\tilde{\bm{\nu}}_t(\theta)=\left({\bm{\nu}^1_t(\theta)}^\top,\cdots,{\bm{\nu}^p_t(\theta)}^\top\right)^\top$. Similarly, $\bm{L}^l_t(\theta)=\left(L^l_{1,t}(\theta),\cdots,L^l_{n,t}(\theta)\right)^\top$ and $\tilde{\bm{L}}_t(\theta)=\left({\bm{L}^1_t(\theta)}^\top,\cdots,{\bm{L}^p_t(\theta)}^\top\right)^\top$. Additionally, denote the matrix
\begin{equation*}
\tilde{A}=\begin{bmatrix}
\gamma_1A & \gamma_2 I & \cdots & \gamma_p I \\
\gamma_1I & \gamma_2A & \cdots & \gamma_p I \\
\vdots & \vdots & \ddots  & \vdots \\
\gamma_1I & \gamma_2 I & \cdots & \gamma_pA
\end{bmatrix},
\end{equation*}
where $I$ is the identity matrix. It is evident that $\tilde{A}$ is a row-stochastic matrix, we further demonstrate that $\tilde{A}$ serves as the transition matrix for an irreducible, aperiodic Markov chain with finite states. Given that $A$ has at least one positive diagonal element, $\tilde{A}$ contains a minimum of $n$ positive diagonal elements, making it aperiodic.

We then examine the strong connectivity of the corresponding graph of the $np\times np$ matrix $\tilde{A}$ to prove its irreducibility. For any node with index $i = n(l-1) + i_0$, where $1 \leq i_0 \leq n$ and $1 \leq l \leq p$, it can establish a path to any node with index $j$ within the range $n(l-1)+1 \leq j \leq nl$ due to the irreducibility of $A$. We need to further find a path from node $i$ to any node $j$, where $j= n(k-1) + j_0$, $1 \leq j_0 \leq n$, and $k \neq l$. Since $\tilde{a}_{n(k-1)+i_0,i} = \left[\gamma_l I\right]_{i_0i_0} = \gamma_l > 0$, a path $P_1$ from node $i$ to node $n(k-1) + i_0$ exists. Additionally, node $n(k-1) + i_0$ can find a path $P_2$ to node $j$ as previously demonstrated. Combining these paths as $P = P_1 \cup P_2$ establishes a path from node $i$ to node $j$. Thus, we can conclude that every node $i$ has a path to any node $j$ in the graph corresponding to matrix $\tilde{A}$, confirming its irreducibility.

Subsequently we can rewrite \eqref{1} in matrix form:
\begin{equation*}
\tilde{\bm{\nu}}_{t+1}(\theta)=\tilde{A}\tilde{\bm{\nu}}_{t}(\theta)+\tilde{\bm{L}}_{t+1}(\theta).
\end{equation*}
Now it follows that
\begin{equation}\label{2}
\begin{aligned}
\frac{1}{t}&\tilde{\bm{\nu}}_{t+1}(\theta)=\frac{1}{t}\tilde{A}\tilde{\bm{\nu}}_{t}(\theta)+\frac{1}{t}\tilde{\bm{L}}_{t+1}(\theta)=\cdots\\
&=\frac{1}{t}\tilde{A}^{t+1}\tilde{\bm{\nu}}_{0}(\theta)+\frac{1}{t}\sum_{k=1}^t\tilde{A}^k\tilde{\bm{L}}_{t+1-k}(\theta)+\frac{1}{t}\tilde{\bm{L}}_{t+1}(\theta).
\end{aligned}
\end{equation}
The assumptions admit that the first and the third terms on r.h.s. of (\ref{2}) go to zero as $t\rightarrow\infty$. The second term can be deformed as
\begin{equation}\label{3}
\begin{aligned}
\frac{1}{t}\sum_{k=1}^t\tilde{A}^k\tilde{\bm{L}}_{t+1-k}(\theta)&=\frac{1}{t}\sum_{k=1}^t(\tilde{A}^k-\bm{1}_{np}\tilde{\pi})\tilde{\bm{L}}_{t+1-k}(\theta)\\
&+\frac{1}{t}\sum_{k=1}^t\bm{1}_{np}\tilde{\pi}(\tilde{\bm{L}}_{t+1-k}(\theta)-\tilde{\bm{K}}(\theta^*,\theta))\\
&+\frac{1}{t}\sum_{k=1}^t\bm{1}_{np}\tilde{\pi}\tilde{\bm{K}}(\theta^*,\theta),
\end{aligned}
\end{equation}
where $\bm{1}_{np}$ is an $np$-dimensional column vector of ones, $\tilde{\pi}$ is the eigenvector centrality corresponding to matrix $\tilde{A}$ and is a row vector, $\tilde{\bm{K}}(\theta^*,\theta)=\left({\bm{K}^1(\theta^*,\theta)}^\top,\cdots,{\bm{K}^p(\theta^*,\theta)}^\top\right)^\top$, and $\bm{K}^l(\theta^*,\theta)=\left(K^l_{1}(\theta^*,\theta),\cdots,K^l_{n}(\theta^*,\theta)\right)^\top$. Lemma 1 admits that $\lim\limits_{k\rightarrow\infty}\tilde{A}^k=\bm{1}_{np}\tilde{\pi}$. Noticing that all elements of $\tilde{A}^k (k=1,2,\cdots)$ are bounded, the first term on r.h.s. of (\ref{3}) converges to zero as $t\rightarrow\infty$. Moreover, under Circumstance 1, for all $l=1,\cdots,p$  we have
\begin{equation*}
\begin{aligned}
\mathbb{E}^*[L^l_{i,t}(\theta)]&=\mathbb{E}^*\left[\log\frac{\ell^l_i(s^l_{i,t}|\theta)}{\ell^l_i(s^l_{i,t}|\theta^*)}\right]\\
&=\int\limits_{s^l\in S^l_i}\ell^l_i(s^l|\theta^*)\log\frac{\ell^l_i(s^l|\theta)}{\ell^l_i(s^l|\theta^*)}{\rm d}s^l\\
&=-D_{\textrm{KL}}(\ell^l_i(\cdot|\theta^*)\parallel\ell^l_i(\cdot|\theta))=K^l_i(\theta^*,\theta).
\end{aligned}
\end{equation*}
The Kolmogorov's strong law of large numbers gives that $\forall l=1,\cdots,p$,
\begin{equation*}
\frac{1}{t}\sum_{k=1}^t\bm{L}^l_{t+1-k}(\theta)-\frac{1}{t}\sum_{k=1}^t\mathbb{E}^*[\bm{L}^l_{t+1-k}(\theta)]\rightarrow \mathbf{0},\quad\mathbb{P}^*{\rm-a.s.},
\end{equation*}
as $t\rightarrow\infty$, which leads to
\begin{equation*}
\lim_{t\rightarrow\infty}\frac{1}{t}\sum_{k=1}^t\bm{1}_{np}\tilde{\pi}(\tilde{\bm{L}}_{t+1-k}(\theta)-\tilde{\bm{K}}(\theta^*,\theta))=\mathbf{0},\quad\mathbb{P}^*{\rm-a.s.}.
\end{equation*}
Now (\ref{3}) gives that
\begin{equation*}
\lim_{t\rightarrow\infty}\frac{1}{t}\sum_{k=1}^t\tilde{A}^k\tilde{\bm{L}}_{t+1-k}(\theta)=\bm{1}_{np}\tilde{\pi}\tilde{\bm{K}}(\theta^*,\theta),\quad\mathbb{P}^*{\rm-a.s.}.
\end{equation*}
According to Assumption 3 and Lemma 1, for all $\theta\neq\theta^*$ we have
\begin{equation}\label{4}
\lim\limits_{t\rightarrow\infty}\frac{1}{t}\tilde{\bm{\nu}}_{t+1}(\theta)<\mathbf{0},\quad \mathbb{P}^*{\rm-a.s.}.
\end{equation}
Thus $\nu^l_{i,t+1}(\theta)\rightarrow-\infty$ almost surely for all agents $i=1,\cdots,n$ and signal types $l=1,\cdots,p$. This implies $\mu^l_{i,t}(\theta)\rightarrow 0$ for all $i=1,\cdots,n$ and $l=1,\cdots,p$ almost surely.
\end{proof}

In the proof of Theorem 1, it is noteworthy that we construct a new row-stochastic matrix $\tilde{A}$ and demonstrate its primitivity. Hence, our proposed algorithm can be viewed as duplicating the network $\mathcal{G}$ of agents into $p$ identical subnetworks $\mathcal{G}^1,\cdots,\mathcal{G}^p$, establishing bidirectional links between each node and its duplicate, and assigning a distinct signal to each subnetwork for classic non-Bayesian social learning with geometric averaging. The weight matrix corresponding to the augmented network is exactly $\tilde{A}$. An illustration is shown in Fig. \ref{illustration}.

\begin{figure}[t]
\centering
\includegraphics[width=1\linewidth]{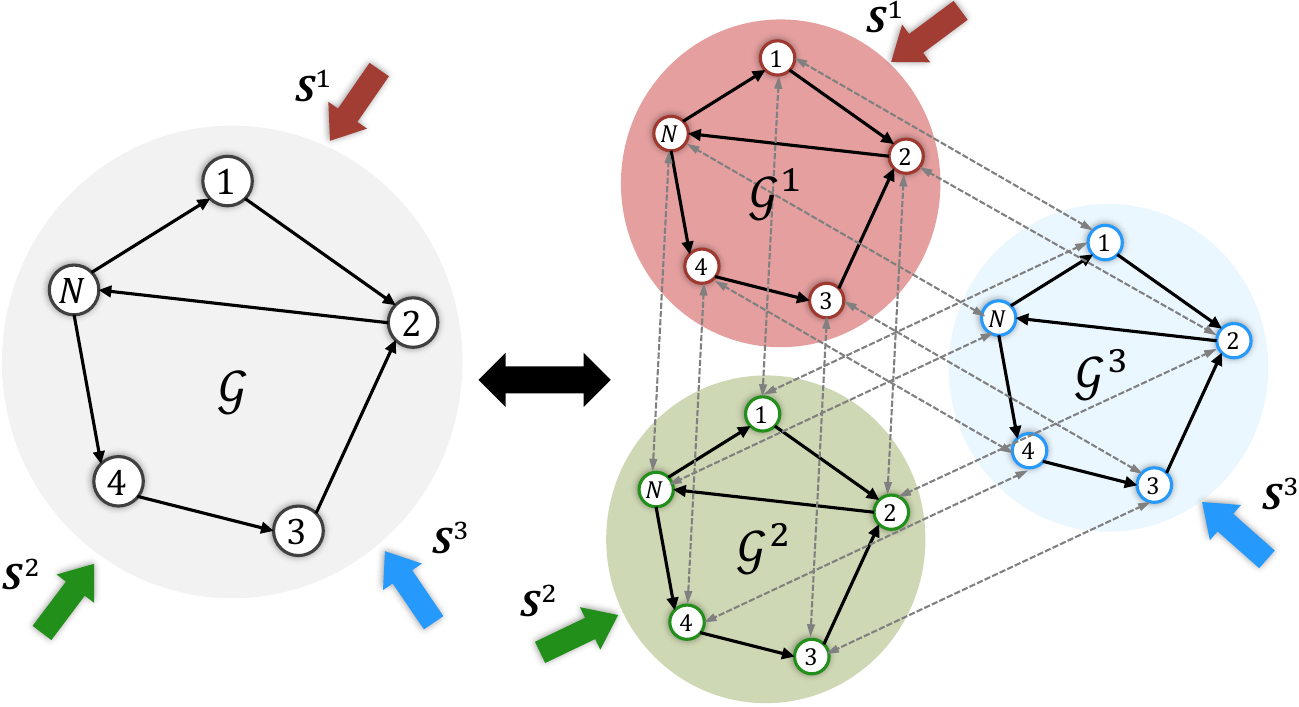}
\caption{An intuitive illustration for an understanding of the proposed algorithm.}\label{illustration}
\end{figure}

Theorem 1 guarantees that all agents will eventually learn the underlying true state with our learning strategy as long as some fundamental assumptions are satisfied. Notably, our algorithm does not require every type of signal to be informative for the group. As long as, for every pair of states $\theta$ and $\theta^*$, there exist an agent capable of distinguishing between them with a certain type of signal, the entire group can achieve correct learning. Subsequently, we will demonstrate that, in certain scenarios, our method is capable of learning correctly even in the presence of misleading signals.

From the proof of Theorem 1, as long as \eqref{4} is satisfied, all agents will learn the underlying true state correctly. Therefore, it is important to figure out $\tilde{\pi}$, leading to the following lemma.
\begin{lemma}
Let $\tilde{\pi}$ be the normalized left eigenvector of matrix $\tilde{A}$ associated with eigenvalue 1, and $\pi$ be the normalized left eigenvector of matrix $A$ associated with eigenvalue 1. Then, we have $\tilde{\pi}=\left(\gamma_1\pi,\cdots,\gamma_p\pi\right)$.
\end{lemma}
\begin{proof}
\begin{equation*}
\begin{aligned}
\tilde{\pi}\tilde{A}=&(\gamma_1(\gamma_1\pi A+\gamma_2\pi+\cdots+\gamma_p\pi),\cdots,\\ &\quad\gamma_p(\gamma_1\pi+\gamma_2\pi+\cdots+\gamma_p\pi A))\\
=&\left(\gamma_1\pi,\cdots,\gamma_p\pi\right)=\tilde{\pi}.
\end{aligned}
\end{equation*}
The second equality follows from $\pi A=\pi$ and $\sum\limits_{l=1}^p\gamma_l=1$.
\end{proof}
\begin{theorem}
Under Circumstance 2 and Assumptions 1 and 2, all agents will correctly learn the underlying true state, i.e.,
\begin{equation*}
\lim_{t\rightarrow\infty}\mu^l_{i,t}(\theta^*)=1,\quad\mathbb{P}^*{\rm-a.s.},\quad\forall 1\le i\le n, 1\le l\le p,
\end{equation*}
if and only if
\begin{equation*}\label{main}
\sum\limits_{l=1}^p\gamma_l\sum\limits_{i=1}^n\pi_iK_i^l(\theta^*,\theta)<0.
\end{equation*}
\end{theorem}
\begin{proof}
Firstly, similar to the proof of Theorem 1, under Circumstance 2 we have
\begin{equation*}
	\begin{aligned}
	\mathbb{E}^*\left[L^l_{i,t}(\theta)\right]&=\mathbb{E}^*\left[\log\frac{\ell^l_i(s^l_{i,t}|\theta)}{\ell^l_i(s_{i,t}|\theta^*)}\right]\\&=\int\limits_{s^l\in S^l_i}f^l_i(s^l)\log\frac{\ell^l_i(s^l|\theta)}{\ell^l_i(s^l|\theta^*)}\mathrm{d}s^l\\
	&=\int\limits_{s^l\in S^l_i}f^l_i(s^l)\left(\log\frac{f^l_i(s^l)}{\ell^l_i(s^l|\theta^*)}-\log\frac{f^l_i(s^l)}{\ell^l_i(s^l|\theta)}\right)\mathrm{d}s^l\\
	&=D_{\textrm{KL}}(f^l_i(\cdot)||\ell^l_i(\cdot|\theta^*))-D_{\textrm{KL}}(f^l_i(\cdot)||\ell^l_i(\cdot|\theta))\\&=K^l_i(\theta^*,\theta).
	\end{aligned}
\end{equation*}
and
\begin{equation}\label{5}
\begin{aligned}
\lim\limits_{t\rightarrow\infty}\frac{1}{t}\tilde{\bm{\nu}}_{t+1}(\theta)&=\bm{1}_{np}\tilde{\pi}\tilde{K}(\theta^*,\theta)\\
&=\bm{1}_p \otimes \sum_{l=1}^p\gamma_l\bm{1}_n\bm{\pi}\bm{K}^l(\theta^*,\theta),
\end{aligned}
\end{equation}
\begin{equation*}
\begin{aligned}
\bm{1}_n\bm{\pi}\bm{K}^l(\theta^*,\theta)=\sum\limits_{i=1}^n\pi_iK_i^l(\theta^*,\theta)\bm{1}_n,
\end{aligned}
\end{equation*}
hence the condition $\sum\limits_{l=1}^p\gamma_l\bm{1}_n\bm{\pi}\bm{K}^l(\theta^*,\theta)<\mathbf{0}$
 holds if and only if $\sum\limits_{l=1}^p\gamma_l\sum\limits_{i=1}^n\pi_iK_i^l(\theta^*,\theta)<0$.
\end{proof}
Theorem 2 demonstrates the robustness of our algorithm when dealing with multiple signals. Even if some types of signals might be misleading, the likelihood of collective mislearning can be reduced by adjusting the assigned parameter $\gamma_l$. Additionally, as can be seen from \eqref{5}, assigning a higher weight $\gamma_l$ to signals that are more instructive, i.e., better able to help the group distinguish between correct and incorrect states, can accelerate the convergence rate of the algorithm.
\section{Numerical examples}

\subsection{Learning with multiview observations}

We first demonstrate that our proposed multiview observations algorithm can address the observational equivalence issue present when only a single viewpoint of information is available.

\begin{example}
Consider a strongly-connected network consisting of two agents. The corresponding weight matrix is
\begin{equation*}
A=\begin{bmatrix}
0 & 1\\
0.7 & 0.3
\end{bmatrix}.
\end{equation*}
The two agents are engaged in the task of localizing a target situated within a $4\times4$ grid. They receive Gaussian signals (type 1) with mean values corresponding to the distances from the target, which could be at any of the 16 grid points. The signal structure of the two agents with respect to $\theta$ also follows a Gaussian distribution, with the mean value equal to the true distance between the agent and $\theta$. As shown in Fig.~\ref{illustration_1}, the two agents struggle to distinguish between two states due to the overlapping circles centered on these states, where each circle has a radius equal to the distance from the target.

At the same time, both agents can receive signals (type 2) regarding whether the target is above or below them. If the target is above an agent, at each moment, there is a 0.8 probability of receiving signal U and a 0.2 probability of receiving signal D. The signal structure is set as $\ell^2_i(U|\theta)=0.8$ if $\theta$ is located above $i$ and $\ell^2_i(U|\theta)=0.2$ if $\theta$ is located below $i$. It is obvious that both of the two agents can not identify the underlying true state based solely on type 2 signal.
\begin{figure}[t]
\centering
\subfigure[]{\includegraphics[width=0.45\linewidth]{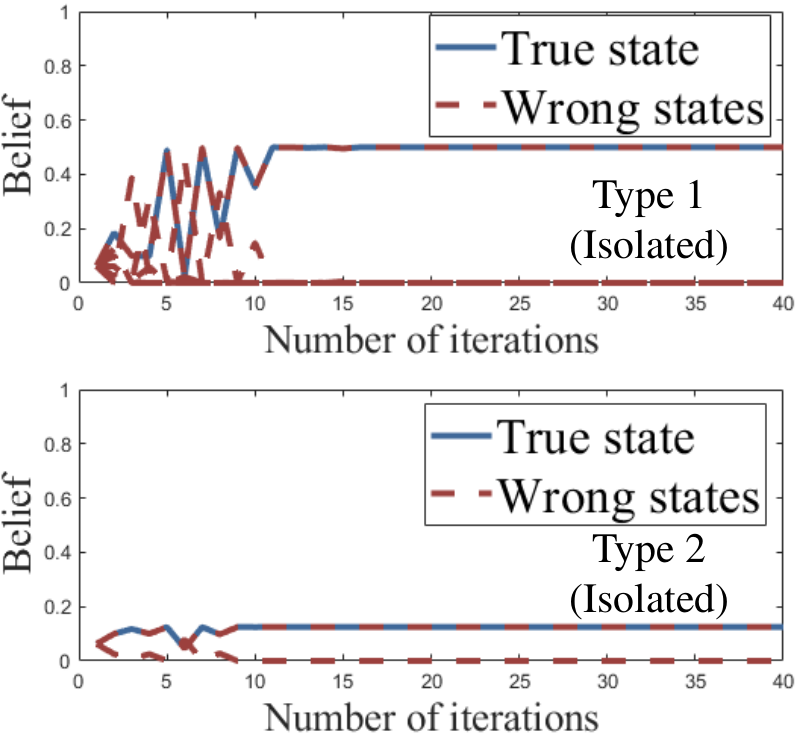}}\hspace{5mm}
\subfigure[]{\includegraphics[width=0.45\linewidth]{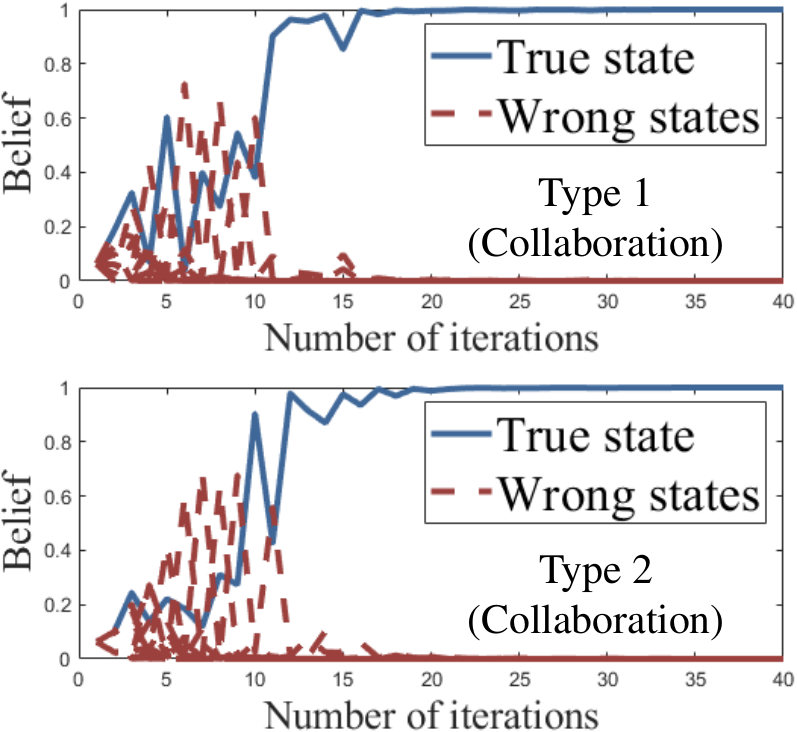}}
\caption{The evolution of beliefs of Agent 1 on different states. (a) The two agents are unable to identify the underlying true state with a single type of signal. (b) The two agents achieve correct learning by combining the information from two types of signals.}\label{fig1}
\end{figure}

\begin{figure}[t]
\centering
\includegraphics[width=0.8\linewidth]{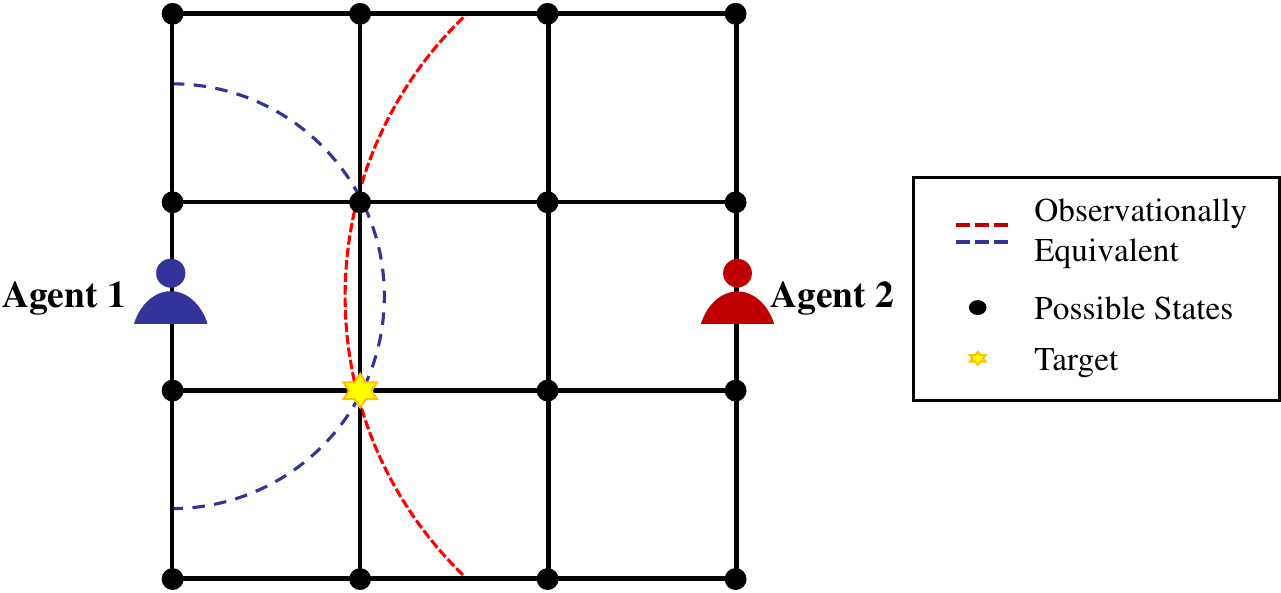}
\caption{Illustration of the scenario in Example 1. In this example, the two agents, due to the observational equivalence problem, cannot achieve correct learning relying solely on a single type of signal.}\label{illustration_1}
\end{figure}
\end{example}
The initial beliefs about both types of signal are uniform distributions over all possible states. Under our settings, the signal structures of all agents and all types of signal about the true state $\theta^*$ align with its actual distribution, which satisfy the conditions in Circumstance 1.

The experimental results indicate that the agents are unable to collectively learn the true state according to one type of signal solely, as shown in Fig.~\ref{fig1} (a). However, by combining the information provided by both types of signals, the two agents can successfully collaborate to identify the underlying true state due to the fact that Assumption 3 is satisfied, as shown in Fig.~\ref{fig2} (b). This example demonstrates that learning from multiple signals can resolve the issue of observationally equivalence present in single-signal scenarios, thereby offering more tolerant conditions for successful learning.

\begin{figure}[t]
\centering
\subfigure[]{\includegraphics[width=0.45\linewidth]{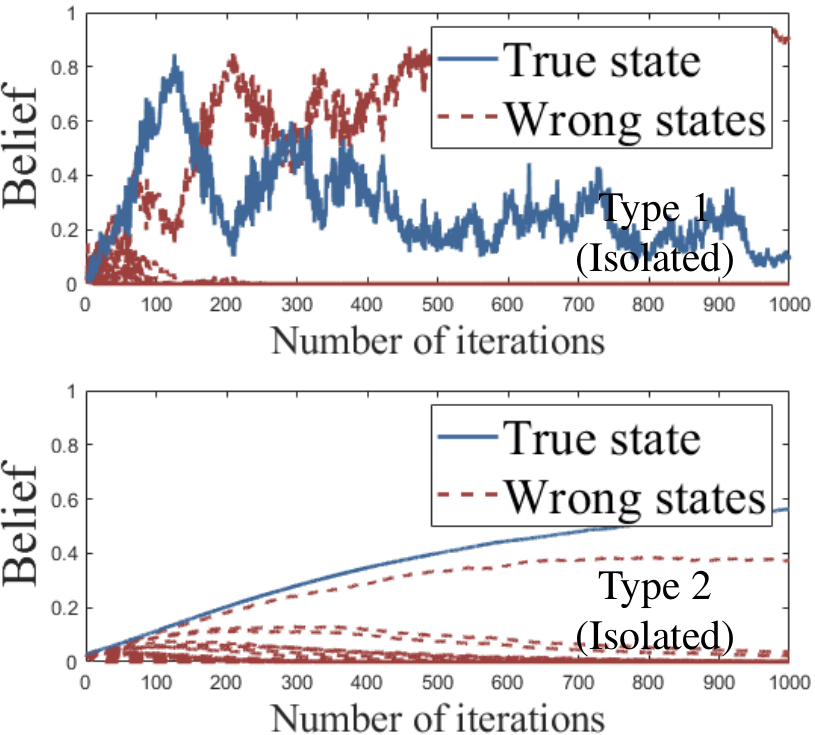}}\hspace{5mm}
\subfigure[]{\includegraphics[width=0.45\linewidth]{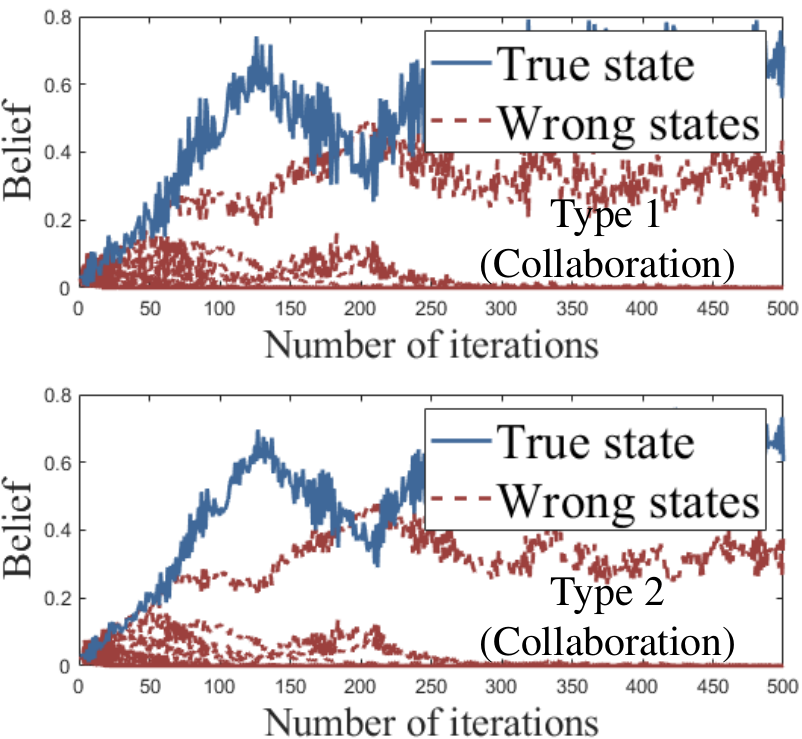}}
\caption{The evolution of beliefs of Agent 1 on all possible states in the first scenario of distributed cooperative localization task. (a) The agent can identify the optimal state solely based on azimuth information, but using only distance information results in erroneous learning. (b) By employing our algorithm to integrate the information from both types of signals, the beliefs of the agent converge to the true state.}\label{fig2}
\end{figure}
\subsection{Learning with misleading signal}

In the following example, we will demonstrate how our proposed algorithm addresses the issue of erroneous learning that may occur with a single viewpoint of signal by aggregating information from multiview observations, thereby validating the enhanced robustness of our algorithm.

\begin{example}
Consider a scenario where a group of $N$ sensors is randomly distributed in a two-dimensional square, denoted as $\left[0,1\right]^2$. Each sensor receives two viewpoints of observations at every time step, one is related to distance and the other related to orientation. The first type is a Gaussian signal with noise, representing the distance from the sensor to the target. Specifically, $\boldsymbol{S}^1_{i,t}\sim\mathcal{N}(d_i^*,\sigma_1^2)$ for all $i=1,\cdots,N$ and $t=1,2,\cdots$, where $d_i^*$ denotes the distance from the $i$-th sensor to the target. The second type of signal also contains Gaussian noise and pertains to the azimuth between the agent and the target, and $\boldsymbol{S}^2_{i,t}\sim\mathcal{N}(\alpha_i^*,\sigma_1^2)$, where $\alpha_i^*$ denotes the azimuth. There are $M$ possible positions $\theta_m$ ($m=1,\cdots,M$) uniformly distributed in the square space. Let $d_i^m$ and $\alpha_i^m$ denote the distance and azimuth from sensor $i$ to a possible state $\theta_m$ respectively. The collective aim of these sensors is to find the state that has a position closest to the target, and it is conceived as underlying true state $\theta^*$. For every agent $i$, we set $\ell^1_i(\cdot|\theta_m)$ to align with $\mathcal{N}(d_i^m,\sigma_1^2)$, and since the orientation serves as auxiliary information, we let $\ell^2_i(\cdot|\theta_m)\sim\mathcal{N}(\alpha_i^m,10\sigma_2^2)$.

We set $N = 2$, $M=36$, $\sigma_1=\sigma_2=0.5$, and $\gamma_1=\gamma_2=0.5$. The weight matrix is the same as that in Example 1, indicating that the corresponding network is strongly-connected. The initial beliefs about both types of signal are uniform distributions over all possible states. In situations where the target does not coincide with any possible state, our experimental setup clearly meets the conditions of Circumstance 2. In this case, based on the proof of Theorem 2 and through calculation, the collective learning outcome $\hat{\theta}$ satisfies:
\begin{equation*}
\begin{aligned}
&\hat{\theta}=\theta_{j_0},
\\&j_0=\underset{j=1,\cdots,m}{\arg\min}\sum\limits_{i=1}^2\pi_i\left((d_i^*-d_i^j)^2+\frac{1}{10}(\alpha_i^*-\alpha_i^j)^2\right),
\end{aligned}
\end{equation*}
where $\boldsymbol{\pi}=(\pi_1,\pi_2)$ is the eigenvector centrality corresponding to the weight matrix.
\begin{figure}[t]
\centering
\subfigure[]{\includegraphics[width=0.45\linewidth]{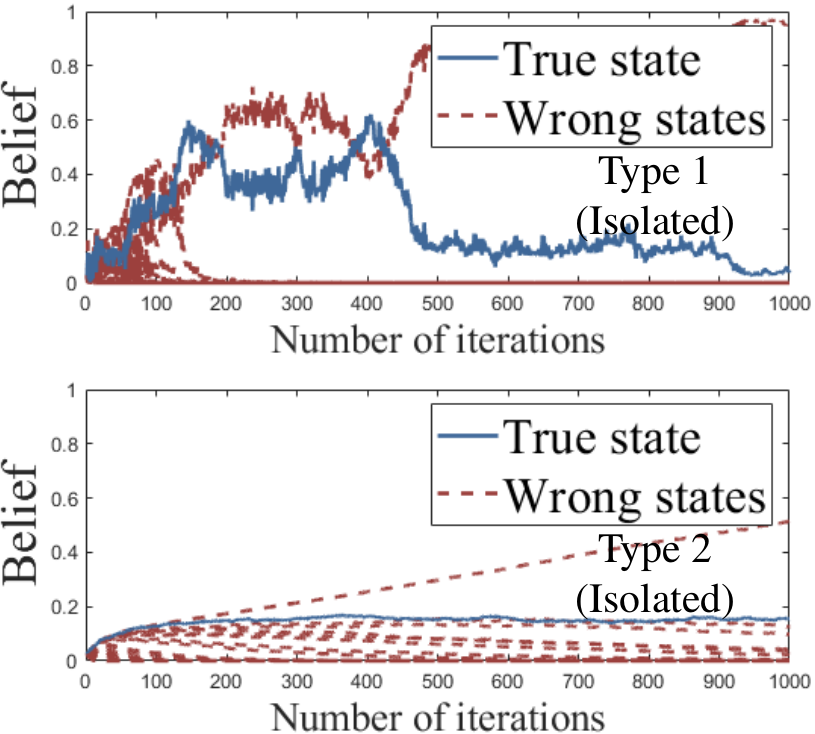}}\hspace{5mm}
\subfigure[]{\includegraphics[width=0.45\linewidth]{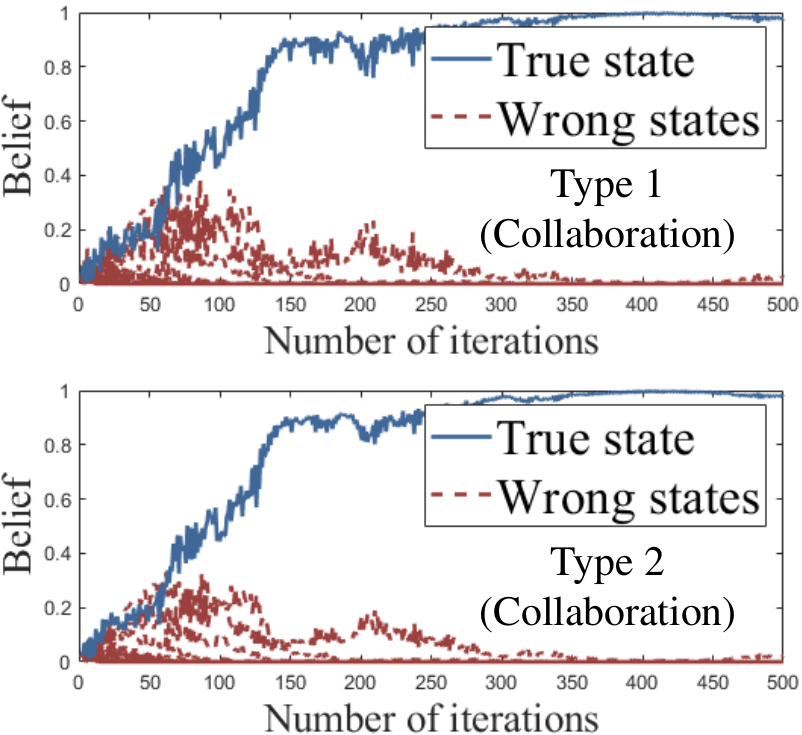}}
\caption{The evolution of beliefs of Agent 1 on all possible states in the second scenario of distributed cooperative localization task. (a) The agents cannot achieve correct learning solely relying on distance or azimuth information. (b) By employing our algorithm to integrate the information from both types of signals, the agents can learn the underlying true state asymptotically.}\label{fig3}
\end{figure}

\begin{figure}[t]
\centering
\includegraphics[width=1\linewidth]{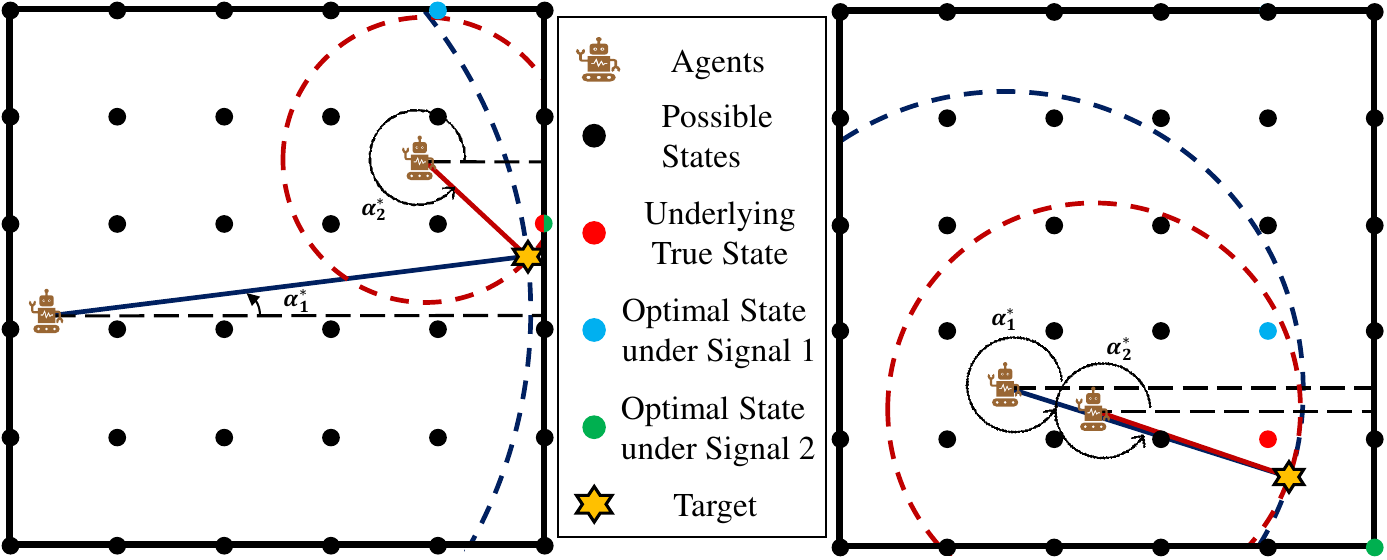}
\caption{Illustrations of two typical scenarios. In the first scenario, the group incorrectly selects the optimal state based on distance information. In the second scenario, due to the similar orientation information of the two agents, no type of signals leads to correct learning.}\label{illustration_2}
\end{figure}
We fix the position of the target and allow agents to be distributed in different locations to study the collaborative learning outcomes with two types of signals separately and jointly. In 1000 experiments, our proposed method, which combines distance and azimuth information, successfully achieves localization 853 times, while relying solely on distance or azimuth information leads to successful localization 550 and 582 times, respectively. This clearly demonstrates the benefits of integrating both types of signal information.

We will further present the detailed results of two typical scenarios. In the first scenario, agents relying solely on distance information fail to find the state closest to the target, while exclusive reliance on azimuth information results in successful localization, as shown in Fig.~\ref{fig2} (a). If they combine both types of information, all agents could successfully achieve the task, as shown in Fig.~\ref{fig2} (b). In the second scenario, agents relying solely on either distance or azimuth information are unable to achieve accurate localization, as evidenced in Fig.~\ref{fig3} (a). However, after applying our method, agents are able to successfully identify the state closest to the target, as demonstrated in Fig.~\ref{fig3} (b). The schematic diagrams of agent and state positions in two scenarios are presented in Fig.~\ref{illustration_2}.

This experiment, serving as a complement to Theorem 2, demonstrates that our proposed multiview observations algorithm can, to a certain extent, alleviate the impact of misleading signals, thereby increasing the fault tolerance of collective learning.

\end {example}

\section{Conclusion and future work}
In this paper, we extend traditional non-Bayesian social learning algorithms designed for single signal and propose a distributed information processing algorithm that integrates information from multiview observations. Our proposed algorithm enables the group to learn from multiple viewpoints of information independently and achieving interaction among multiview observations during the information aggregation step. By introducing weight parameters for various signal types, we not only ensure the convergence of the algorithm under traditional assumptions, but also, in certain scenarios, correct errors introduced by a single view of signal, significantly enhancing the fault tolerance of collective learning.

Our work not only presents a distributed information processing algorithm capable of handling a more diverse range of task scenarios but also contributes to the theoretical foundation of distributed machine learning. In future endeavors, we plan to explore the application of non-Bayesian social learning algorithms with multiview observations in the design of distributed machine learning methods to address challenges associated with multi-feature or high-dimensional problems. Additionally, we may consider other interaction mechanisms among multiple signals, such as negative feedback, and introduce more parameter settings to enable applications in a wider range of tasks.

\bibliographystyle{IEEEtran}
\bibliography{IEEEabrv,main}

\end{document}